\documentclass{PoS}
\usepackage{amsmath, amsthm, amsfonts, amssymb, bbm}

\theoremstyle{plain}
\newtheorem{proposition}{Proposition}

\theoremstyle{definition}
\newtheorem{definition}{Definition}

\theoremstyle{remark}
\newtheorem{remark}{Remark}

\newcommand{\rhs}{r.h.s.\ }

\newcommand{\wrt}{w.r.t.\ }

\newcommand{\ud}{\mathrm{d}}
\newcommand{\del}{\partial}
\newcommand{\Star}{{\star}}

\newcommand{\GS}{\mathcal{S}}
\newcommand{\R}{\mathbb{R}}
\newcommand{\betrag}[1]{\lvert #1 \rvert}
\newcommand{\skal}[2]{\langle #1 , #2 \rangle}

\title{Divergences in QFT on the Noncommutative Minkowski Space with
Grosse-Wulkenhaar potential}

\ShortTitle{QFT on the Noncommutative Minkowski Space with
Grosse-Wulkenhaar potential}

\author{\speaker{Jochen Zahn}%
\\
	Institut f\"ur Theoretische Physik\\
        Leibniz Universit\"at Hannover, Appelstra{\ss}e 2, 30167 Hannover, Germany\\
        E-mail: \email{jzahn@uni-math.gwdg.de}}

\abstract{We study quantum field theory on the two-dimensional Noncommutative Minkoswki space with a Grosse-Wulkenhaar potential. We explicitly construct the retarded propagator and show that it is not a tempered distribution. This leads to problems when trying to define planar products of such distributions, as they appear in the Yang-Feldman series. At and above the self-dual point, these can no longer be defined, not even at different points. This shows that we do not deal with an ordinary ultraviolet divergence.}

\FullConference{Corfu Summer Institute on Elementary Particles and Physics -
Workshop on Non Commutative Field Theory and Gravity\\
                 September 8 - 12, 2010\\
                 Corfu, Greece}

\begin{document}

\section{Introduction}
In quantum field theory on the Moyal space with Euclidean signature, the introduction of the Grosse-Wulkenhaar potential was a spectacular success. With this potential, the noncommutative $\phi^4$ model is renormalizable not only in two~\cite{GrosseWulkenhaar2d}, but also in four spacetime dimensions~\cite{GrosseWulkenhaar4d}. Even better, the model is asymptotically safe (but not free!), since the $\beta$--function is bounded~\cite{GhostGW,GhostDGMR}.

However, very little is known about the Minkowski space versions of these models. A first step in this direction was done by Fischer and Szabo \cite{FischerSzabo}, who defined the free Grosse-Wulkenhaar model on Minkowski space to be given by the Lagrangean
\begin{equation}
\label{eq:Lagrangean}
 L = \phi \Box \phi - \Omega^2 (2 \sigma^{-1} x)^2 \phi^2 + m^2 \phi^2.
\end{equation}
Here $\sigma$ is the real antisymmetric matrix that defines the $\Star$-product through
\begin{equation}
\label{eq:StarProduct}
(f \Star h){\hat \ }(k) = (2\pi)^{-d/2} \int \ud^d l \ e^{-\frac{i}{2} k \sigma l} \hat{f}(k-l) \hat{h}(l),
\end{equation}
where the hat denotes Fourier transformation. This product is used to define interaction terms,
as, e.g., $L_{\mathrm{int}} = g \phi^{\Star 4}$. The real number $\Omega$ plays the role of a coupling constant. The Langrangean \eqref{eq:Lagrangean} is a straightforward generalization of the Lagrangean of Grosse and Wulkenhaar, the only modification being that Euclidean products are replaced by Lorentzian ones in the first and the second term. As in the Euclidean case, the Lagrangean \eqref{eq:Lagrangean} has a remarkable symmetry property in the case $\Omega = 1$. The model is then symmetric under a certain exchange of positions and momenta \cite{LangmannSzabo,GhostGW}. Hence, one calls $\Omega = 1$ the self-dual point. This value is also special for another reason, since it is the fixed point for the $\phi^{\Star 4}$ model in four dimensions \cite{GhostGW,GhostDGMR}.

From the definition of the Lagrangean \eqref{eq:Lagrangean}, it is clear that $\sigma$ must be invertible. Thus, the time direction will not commute with all the spatial directions. In such a setting, quantization has to be done with care. In general, the naive application of Feynman rules derived in the Euclidean setting leads to a violation of unitarity \cite{GomisMehen}. As an approach that does not share this problem, quantization in the Yang-Feldman formalism was proposed in \cite{BDFP02}, and this is the way we proceed here.

As we will see, in the quantum theory thus defined, the self-dual point is also special, but in a less desirable way: In the two-dimensional case, the planar products of propagators that appear in the Yang-Feldman series are ill-defined at and above the self-dual point. The reason is that the retarded propagator grows rapidly (as a Gaussian) in some directions, meaning that appropriate test functions have to be strongly localized in these directions. But these directions do not commute under this product, so that the localization can not be maintained under the planar product. As the difficulty already arises before taking the limit of coinciding points, this is not an ordinary ultraviolet divergence.

We proceed as follows: We recall the setup of the Yang-Feldman series and show how planar and nonplanar products of propagators arise. We then construct and discuss the retarded propagator, which is a crucial ingredient of the Yang-Feldman formalism. From its asymptotic behavior we conclude that planar products are not defined at and above the self-dual point.

This talk summarizes the results of \cite{GWMink}. The discussion of the appropriate test function spaces in Section~\ref{sec:RetProp} is slightly updated \wrt that article.

\section{The Yang-Feldman formalism}
The idea behind the Yang-Feldman formalism is a perturbative and recursive construction of the interacting field in terms of the free incoming field. As an example, we consider the $\phi^{\Star 3}$ model, i.e., we have the equation of motion
\[ P \phi = g \phi \Star \phi, \]
where $P$ is some wave operator. Now one writes the interacting field as a formal power series in the coupling constant $g$, i.e.,
\[ \phi = \sum_{n=0}^\infty g^n \phi_n. \]
Inserting this ansatz into the equation of motion, one obtains
\begin{equation}
\label{eq:YF_eom}
 P \phi_n = \sum_{k=0}^{n-1} \phi_k \Star \phi_{n-1-k}.
\end{equation}
Thus, $\phi_0$ is the free field. We identify it with the incoming field. Higher order fields are thus obtained by convolution with the retarded propagator:
\begin{align*}
\phi_1(x) & = \int \ud^4y \ \Delta_{\mathrm{ret}}(x,y) \phi_0(y) \Star_y \phi_0(y), \\
\phi_2(x) & = \int \ud^4y \ud^4z \ \Delta_{\mathrm{ret}}(x,y) \left\{ [\phi_0(z) \Star_z \phi_0(z) \Delta_{\mathrm{ret}}(y,z)] \Star_y \phi_0(y) + \phi_0(y) \Star_y [\Delta_{\mathrm{ret}}(y,z) \phi_0(z)  \Star_z \phi_0(z)] \right\}.
\end{align*}
Here we indicated the occurence of $\phi_1$ plugged into the \rhs of \eqref{eq:YF_eom} by square brackets. By this procedure, one obtains tree graphs. Loops occur when one considers contractions, i.e., when a product of fields is replaced by the two-point function $\Delta_+$. For $\phi_2$, there are four possible contractions between a field at $y$ and a field at $z$ (contractions between fields at $z$ are tadpoles, which we do not consider). This yields
\begin{multline*}
\phi_2(x) = \int \ud^4y \ud^4z \ \Delta_{\mathrm{ret}}(x,y) \phi_0(z) \left\{ \Delta_{\mathrm{ret}}(y,z) \Star_y \bar{\Star}_z \Delta_+(z,y) + \Delta_{\mathrm{ret}}(y,z) \Star_y \Star_z \Delta_+(z,y) \right. \\
 + \left. \Delta_+(y,z) \Star_y \bar{\Star}_z \Delta_{\mathrm{ret}}(y,z) + \Delta_+(y,z) \Star_y \Star_z \Delta_{\mathrm{ret}}(y,z) \right\}
\end{multline*}
plus the uncontracted part. Here $\bar{\Star}$ is defined by $f \bar{\Star} h = h \Star f$, i.e., with $\sigma$ replaced by $-\sigma$. The expression in curly brackets is what is usually called the self-energy. In terms of the planar and the nonplanar product on $\R^{2d}$,
\begin{align*}
(f \Star_{\mathrm{pl}} h) (y,z) & = f(y,z) \Star_y \bar{\Star}_z h(y,z), \\
(f \Star_{\mathrm{np}} h) (y,z) & = f(y,z) \Star_y \Star_z h(y,z),
\end{align*}
we thus find the one-loop self-energy
\begin{multline}
\label{eq:SelfEnergy}
\Sigma(y,z) = \Delta_{\mathrm{ret}}(y,z) \Star_{\mathrm{pl}} \Delta_-(y,z) + \Delta_{\mathrm{ret}}(y,z) \Star_{\mathrm{np}} \Delta_-(y,z) \\ + \Delta_+(y,z) \Star_{\mathrm{pl}} \Delta_{\mathrm{ret}}(y,z) + \Delta_+(y,z) \Star_{\mathrm{np}} \Delta_{\mathrm{ret}}(y,z),
\end{multline}
where we introduced the notation $\Delta_-(y,z) = \Delta_+(z,y)$. We thus see that in order to do quantum field theory on Moyal space with Lorentzian signature in the Yang-Feldman approach, we have to be able to define planar and nonplanar products of $\Delta_{\mathrm{ret}}$ and $\Delta_{\pm}$.


\section{The retarded propagator}
\label{sec:RetProp}
In the preceeding section, we saw that the retarded propagator is an essential ingredient of the Yang-Feldman formalism. Furthermore, it is uniquely defined, contrary to the Feynman propagator\footnote{Due to the absence of translation invariance, there is no preferred vacuum state.}. Here we explicitly construct it in the massless, two-dimensional case. 

In the absence of the quadratic potential and a mass term, the wave operator for a scalar field is given by $P = \Box = 4 \del_u \del_v$, where we use the light cone coordinates
\begin{equation*}
u = x_0 - x_1, \quad v = x_0 + x_1,
\end{equation*}
The retarded propagator for this wave operator is 
\[
\Delta_{\mathrm{ret}}(u_1, v_1; u_2, v_2) = \tfrac{1}{2} H(u_1-u_2) H(v_1-v_2),
\]
where $H$ is the Heaviside distribution.

In the presence of the quadratic potential, the retarded propagator will no longer be translation invariant, and the above propagator is multiplied with a function of $u_1, v_1, u_2, v_2$. In two dimensions, there is up to multiplication only one real antisymmetric matrix, so that we define
\[
 \sigma = \lambda_{\mathrm{nc}}^2 \epsilon.
\]
In the massless case, the wave operator for the Grosse-Wulkenhaar potential is then given by, cf. \eqref{eq:Lagrangean},
\begin{equation}
\label{eq:WaveOperator_uv}
4 \del_u \del_v + 4 \lambda^{-4} uv,
\end{equation}
where we introduced the notation
\[
 \lambda = \Omega^{-\frac{1}{2}} \lambda_{\mathrm{nc}}.
\]
We have the following proposition, whose straightforward proof can be found in \cite{GWMink}:
\begin{proposition}
The retarded propagator for the wave operator \eqref{eq:WaveOperator_uv} is given by
\begin{align}
\label{eq:Delta_ret}
\Delta_{\mathrm{ret}}(u_1, v_1; u_2, v_2) & = \tfrac{1}{2} H(u_1-u_2) H(v_1-v_2) \sum_{n=0}^\infty (-1)^n \frac{\left( u_1^2 - u_2^2 \right)^n}{2^n \lambda^{2n} n!} \frac{\left( v_1^2 - v_2^2 \right)^n}{2^n \lambda^{2n} n!}  \\
 & = \tfrac{1}{2} H(u_1-u_2) H(v_1-v_2) J_0 \left( \lambda^{-2}\sqrt{(u_1-u_2)(u_1+u_2)(v_1-v_2)(v_1+v_2)} \right). \nonumber
\end{align}
\end{proposition}

\begin{remark}
In a massive theory (without quadratic potential), the retarded propagator is given by
\[
\Delta_{\mathrm{ret}} = \tfrac{1}{2} H(u_1-u_2) H(v_1-v_2) J_0(m \sqrt{(u_1-u_2)(v_1-v_2)}).
\]
Comparison of \eqref{eq:Delta_ret} with this expression shows that effectively we are dealing with a position dependent mass
\[ m^2 = \lambda^{-4} (u_1+u_2)(v_1+v_2). \]
This is the value of the potential at the center of mass of the two points $(u_1, v_1)$ and $(u_2, v_2)$. The problems we will encounter stem from the fact that for $(u_1+u_2) (v_1+v_2) < 0$ the model becomes tachyonic (and ever more so as $(u_1+u_2) (v_1+v_2) \to -\infty$).
\end{remark}

We notice that for imaginary arguments, $J_0$ behaves as 
\[
 J_0(ix) \sim \frac{e^{x}}{\sqrt{2\pi x}}.
\]
Thus, if either $(u_1+u_2)<0$ or $(v_1+v_2)<0$, the retarded propagator grows like a Gaussian. It is thus no tempered distribution. In order to discuss suitable test function spaces on which it can act, we introduce the coordinates
\begin{align*}
 z_1 & = + u_1 - v_2, & z_2 & = - u_2 + v_1, \\
 z_3 & = + u_1 + v_2, & z_4 & = + u_2 + v_1.
\end{align*}
The retarded and the advanced propagator are well defined on Schwartz function that fulfill the bound
\begin{equation}
\label{eq:Exact_Bound}
 \betrag{\del^\beta f} \leq C^\beta e^{- \frac{1+\varepsilon}{4 \lambda^2} \betrag{z_1^2 - z_2^2}}
\end{equation}
for an arbitrarily small positive $\varepsilon$. However, such a bound is hard to control analytically. Thus, we consider Schwartz functions that fulfill the stronger bound
\begin{equation}
\label{eq:GS_Bound}
 \betrag{\del^\beta f} \leq C^\beta e^{- \frac{1+\varepsilon}{4 \lambda^2} (z_1^2 + z_2^2)}.
\end{equation}
This is exactly the definition of a Gelfand-Shilov space:
\begin{definition}
The Gelfand-Shilov space $\GS_{\alpha,A}(\R^4)$ is the space of Schwartz functions that fulfill the bound
\[ \betrag{\del^\beta f(z)} \leq C^\beta e^{-\sum_{i=1}^4 a_i \betrag{z_i}^{1/\alpha_i}}. \]
Here
\[
 a_i = \alpha_i e^{-1} A_i^{-1/\alpha_i} - \delta,
\]
and $\delta > 0$ can be chosen arbitrarily small.
\end{definition}
We thus interpret $\Delta_{\mathrm{ret}}$ and $\Delta_{\mathrm{adv}}$ as distributions on $\GS_{\alpha,A}(\R^4)$, with
\begin{align}
\label{eq:alphaA}
 \alpha & = (\tfrac{1}{2},\tfrac{1}{2},\infty,\infty), &
 A & = \sqrt{2e^{-1}}(\lambda -\varepsilon) (1, 1, \infty, \infty).
\end{align} 
For the discussion in the following section it is a crucial result that the Fourier transforms of such functions are, for $k_3, k_4$ fixed, entire functions whose growth in the imaginary directions is bounded by
\[
\betrag{\hat{f}(k_1+ip_1, k_2+ip_2, k_3, k_4)} \leq C e^{\left( (\lambda-\varepsilon)^2 + \delta \right) \betrag{p}^2},
\]
where $\delta>0$ can be chosen arbitrarily small. The Fourier transformation is a continuous bijection between the spaces $\GS_{\alpha,A}(\R^4)$ and $\GS^{\alpha,A}(\R^4)$.

We remark that the planar and nonplanar commutators for our new coordinates are (all other commutators vanish)
\begin{equation}
\label{eq:Commutators}
 [z_1, z_2]_{\mathrm{pl}} = [z_3, z_4]_{\mathrm{pl}} = [z_1, z_4]_{\mathrm{np}} = [z_3, z_2]_{\mathrm{np}} = 4 i \lambda_{\mathrm{nc}}.
\end{equation}
Thus, we may already anticipate a potential clash with \eqref{eq:GS_Bound}, as our test functions have to be localized in the directions $z_1$ and $z_2$ which do not commute under the planar product. One can expect that such a bound is not stable under the planar product if $\lambda$ is too small.

\section{The planar divergence}

We now want to define the planar and nonplanar $\Star$-products appearing in \eqref{eq:SelfEnergy}. However, even the pointwise (commutative) product is not well-defined in the sense of H\"ormander's product of distributions.
The problems connected to the $\Star$-product and the distributional character of the propagators can be disentagled in the following way\footnote{In \cite{Quasiplanar}, this strategy was pursued for the definition of quasiplanar Wick products.}: In a first step, one defines the $\Star$-product at different points. In the next step, one checks whether the limit of coinciding points makes sense.

The definition of the $\Star$-product of distributions at different points can be done by duality, as proposed in \cite{Soloviev}:
\begin{equation}
\label{eq:Duality}
\skal{\Delta_{\mathrm{ret}} \otimes_{\Star_{\mathrm{pl}}} \Delta_-}{f \otimes h} = \skal{\Delta_{\mathrm{ret}} \otimes \Delta_-}{f \otimes_{\Star_{\mathrm{pl}}} h}.
\end{equation}
As the $\Star$-product at coinciding points \eqref{eq:StarProduct}, the $\Star$-product at different points is best defined in momentum space, i.e., 
\begin{equation}
\label{eq:TwistedProduct}
(f \otimes_{\Star_{\mathrm{pl}}} h) \hat{\ } (k; \tilde{k}) = e^{-2 i \lambda_{\mathrm{nc}}^2 k \sigma_{\mathrm{pl}} \tilde k} \hat{f}(k) \hat{h}(\tilde{k}).
\end{equation}
Using \eqref{eq:Commutators}, we obtain
\[
k \sigma_{\mathrm{pl}} \tilde k = k_1 \tilde{k}_2 - k_2 \tilde{k}_1 + k_3 \tilde{k}_4 - k_4 \tilde{k}_3.
\]
As the retarded propator is an element of $\GS'_{\alpha, A}(\R^4)$, it is natural to expect that the two-point function $\Delta_+$ also lies in a subspace of that space. In order for the \rhs of \eqref{eq:Duality} to be well-defined, we then have to require the \rhs of \eqref{eq:TwistedProduct} to be an element of $\GS^{\alpha \oplus \alpha, A \oplus A}(\R^8)$. For this, we might have to choose $f$ and $h$ from a suitable subset of $\GS_{\alpha, A}(\R^4)$. In \cite{GWMink} it is shown that this is possible if one is far enough below the self-dual point. That this is not possible at or above the self-dual point is shown by the following

\begin{proposition}
For $\lambda_{\mathrm{nc}} \geq \lambda$ and $\alpha$, $A$ as in \eqref{eq:alphaA}, there are no nontrivial $f,h \in \GS_{\alpha, A}(\R^4)$, such that \eqref{eq:TwistedProduct} is the Fourier transform of an element of $\GS_{\alpha \oplus \alpha, A \oplus A}(\R^8)$.
\end{proposition}

\begin{proof}
We fix $k_3, k_4, \tilde{k}_3, \tilde{k}_4$, so that $k=(k_1, k_2)$, $\tilde k = (\tilde k_1, \tilde k_2)$, and analogously for $p$, $\tilde p$. It suffices to show that the bounds (with $b = (\lambda-\epsilon)^2 + \delta$)
\begin{gather}
\label{eq:Ineq1}
\betrag{\hat f(k+ip)} \leq c e^{b \betrag{p}^2}, \quad \betrag{\hat h(k+ip)} \leq c' e^{b \betrag{p}^2}, \\
\label{eq:Ineq2}
\betrag{(f \otimes_{\Star_{\mathrm{pl}}} h) \hat{\ } (k+ip; \tilde{k}+i\tilde{p})} \leq C e^{b ( \betrag{p}^2 + \betrag{\tilde{p}}^2 )},
\end{gather}
can not be fulfilled simultaneously for $\lambda_{\text{nc}}^2 \geq \lambda^2 > b$ ($\delta$ can always be chosen small enough in order to fulfill this inequality). With $l = k + i p$ we have 
\begin{equation*}
\betrag{(f \otimes_{\Star_{\mathrm{pl}}} h) \hat{\ } (l; \tilde{l})} = e^{2 \lambda_\mathrm{nc}^2 ( p \sigma_\mathrm{pl} \tilde k + k \sigma_\mathrm{pl} \tilde p)} \betrag{\hat f(l) \hat h(\tilde l)}.
\end{equation*}
Using \eqref{eq:Ineq2}, we obtain
\[
\betrag{\hat f(l) \hat h(\tilde l)} \leq C e^{-2 \lambda_\mathrm{nc}^2 ( p \sigma_\mathrm{pl} \tilde k + k \sigma_\mathrm{pl} \tilde p)} e^{b ( \betrag{p}^2 + \betrag{\tilde{p}}^2 )}.
\]
Setting $\tilde l = i \sigma_{\mathrm{pl}}^{-1} l$, we thus find
\begin{equation*}
\betrag{ \hat{f}(l) \hat{h}(i \sigma_{\mathrm{pl}}^{-1} l)} \leq C e^{- ( 2 \lambda_{\mathrm{nc}}^2 - b ) \betrag{k}^2 + ( 2 \lambda_{\mathrm{nc}}^2 + b ) \betrag{p}^2}.
\end{equation*}
But from \eqref{eq:Ineq1} we obtain
\begin{equation*}
\betrag{ \hat{f}(l) \hat{h}(i \sigma_{\text{pl}}^{-1} l)} \leq C' e^{b ( \betrag{k}^2 + \betrag{p}^2 )}.
\end{equation*}
Thus, for $\lambda_{\text{nc}}^2 > b$, the entire function $F(l) = e^{bl^2} \hat{f}(l) \hat{g}(i \sigma_{\text{pl}}^{-1} l)$ is bounded on the real and the imaginary axis and grows with order 2 in between. By the Phragm\'en-Lindel\"of principle, it vanishes.
\end{proof}

We have thus shown that at and above the self-dual point the planar $\Star$-product at different points can not be defined via duality on elements of $\GS'_{\alpha,A}(\R^4)$. We interpreted $\Delta_{\mathrm{ret}}$ as an element of that space, so we expect $\Delta_+$ to also lie in a subspace of that space. It follows that the products of distributions appearing in the Yang-Feldman series do not exist, \emph{not even at different points}. In this precise sense, we do not deal with an ultraviolet divergence. We note that if the planar $\Star$-product of two such distributions is not well-defined, then the same is true for higher order products. Thus, the same problems appear for any other polynomial interaction term, in particular also for the $\phi^{\Star 4}$ model. We also note that due to $[x_1, x_2]_{\text{np}} = 0$, no such restrictions occur for nonplanar $\Star$-products.

\begin{remark}
A weak point in the above analysis was the replacement of the bound \eqref{eq:Exact_Bound} by the stronger bound \eqref{eq:GS_Bound}. Thus, one might wonder whether one misses some test functions for which the construction works. However, also a formal direct calculation of the products $\Delta_{\mathrm{ret}} \Star_\mathrm{pl} \Delta_{\mathrm{ret}}$ and $\Delta_{\mathrm{ret}} \Star_\mathrm{pl} \Delta_+$ fails at and above the self-dual point\footnote{In the absence of translation invariance, $\Delta_+$ is not unique, so a reasonable choice has to be taken here.} \cite{GWMink}. One finds (as the coefficient of the $n = 0$ term in the expansion analogous to \eqref{eq:Delta_ret}) the geometric series
\[
\sum_{m=0}^\infty \Omega^{4m},
\]
which clearly diverges at and above the self-dual point.
\end{remark}

\begin{remark}
The above restriction to the massless case is just a matter of convenience. In the massive case, the retarded propagator will no longer be given in a simple analytic form as \eqref{eq:Delta_ret}. However, the quadratic potential will always dominate the mass term sufficiently far away from the origin. Thus, the long distance behavior of the retarded propagator will not change when a mass term is introduced. Hence, the above argument still applies.
\end{remark}

\begin{remark}
\label{rem:Eigenfunctions}
At the self-dual point, there is a natural basis of eigenfunctions of the wave operator that are orthonormal \wrt the $\Star$-product. This is analogous to the harmonic oscillator basis used in the Euclidean case \cite{GrosseWulkenhaar2d,GrosseWulkenhaar4d}, the main difference being that in the present case the basis is continuous, as the spectrum of the wave operator is the whole real line. This difference however, is crucial. It means that in loop calculations the sums have to be replaced by integrals and the Kronecker $\delta$'s at the vertices by Dirac $\delta$'s. In planar subdiagrams this leads to the occurence of squares of Dirac $\delta$'s. Thus, also in this setting, one finds a divergence in planar diagrams at the self-dual point \cite{GWMink}. However, contrary to the approach described above, only the self-dual point can be treated in that way. Furthermore, the localization properties of the eigenfunctions are not known, so it is not clear how to define a retarded propagator with them, and the reason for the occurence of the divergence remains somewhat obscure.
\end{remark}

\section{Conclusion}

We discussed noncommutative field theory with a Grosse-Wulkenhaar potential on the the two-dimensional Minkowski space. We found the retarded propagator for this potential and showed that it grows rapidly in some directions. From its growth rate and the fact that the divergent directions do not commute under the planar $\Star$-product, we concluded that the planar $\Star$-products appearing in the Yang-Feldman series do not exist at and above the self-dual point. As the problem appears already before considering the limit of coinciding points, this is clearly no ultraviolet problem.

In our opinion, the appearance of this new type of divergences is an interesting phenomenon that deserves more detailed studies. Of particular importance is the study of the four-dimensional case, as there the self-dual point is a fixed point of the Euclidean model. A treatment in terms of the eigenfunctions mentioned in Remark~\ref{rem:Eigenfunctions} leads to the same problems as in the two-dimensional case. Preliminary results suggest that this is also true in position space. In that case, it would be important to understand why these problems are absent in the Euclidean setting. If the divergences persist in four dimensions, then one should think about ways to renormalize the model. However, due to the uncommon type of the divergences, it is not clear whether such a program can be successful.


\end{document}